\RequirePackage{amsmath}
\documentclass[runningheads]{llncs}

\usepackage{multicol}
\usepackage{amsfonts}
\usepackage{amssymb}
\usepackage{graphicx}
\usepackage{epic}
\usepackage{eepic}
\usepackage{epsfig,float}
\usepackage{verbatim}
\usepackage{pdfsync}
\usepackage{gastex}
\usepackage{bbm}

\usepackage{xcolor}

\usepackage[utf8]{inputenc}\usepackage[T1]{fontenc}

\renewcommand{\le}{\leqslant}
\renewcommand{\ge}{\geqslant}

\newcommand{\ol}{\overline}
\newcommand{\eps}{\varepsilon}
\newcommand{\emp}{\emptyset}

\newcommand{\Sig}{\Sigma}

\newcommand{\noin}{\noindent}

\newcommand{\bi}{\begin{itemize}}
\newcommand{\ei}{\end{itemize}}
\newcommand{\be}{\begin{enumerate}}
\newcommand{\ee}{\end{enumerate}}
\newcommand{\bd}{\begin{description}}
\newcommand{\ed}{\end{description}}
\newcommand{\bq}{\begin{quote}}
\newcommand{\eq}{\end{quote}}

\newcommand{\cD}{{\mathcal D}}

\newcommand{\cN}{{\mathcal N}}

\newcommand{\cP}{{\mathcal P}}

\newcommand{\cT}{{\mathcal T}}
\newcommand{\cU}{{\mathcal U}}

\newcommand{\one}{{\mathbbm1}}

\title{Unrestricted State Complexity of Binary Operations on  Regular Languages\thanks{This work was supported by the Natural Sciences and Engineering Research Council of Canada 
grant No.~OGP0000871.}
}

\author{Janusz~Brzozowski}

\titlerunning{Unrestricted State Complexity}

\authorrunning{J. Brzozowski}   

\institute{David R. Cheriton School of Computer Science, University of Waterloo, \\
Waterloo, ON, Canada N2L 3G1\\
\{{\tt brzozo@uwaterloo.ca}\}
}

\begin{document}
\maketitle

\begin{abstract}
I study the state complexity of binary operations on regular languages over different alphabets.
It is well known that if $L'_m$ and $L_n$ are languages restricted to be over the same alphabet, with $m$ and $n$ quotients, respectively,  the state complexity of any binary boolean operation on $L'_m$ and $L_n$ is $mn$, and that of the product (concatenation) is $(m-1)2^n +2^{n-1}$. 
In contrast to this, I show that if $L'_m$ and $L_n$ are over their own different alphabets, 
the state complexity of union and symmetric difference is $mn+m+n+1$, that of intersection\footnote{An earlier version of this paper will appear in the Proceedings of DCFS 2016. In that version the complexity of intersection was incorrectly stated to be $mn+1$  because an empty state not reachable by words in the alphabet of the intersection was not removed. For the same reason, the complexity  of difference was incorrectly given as $mn+m+1$. This is explained in detail in the proof of Theorem 1.}  is $mn$,
that of difference is $mn+m$,
and that of the product is  $m2^n+2^{n-1}$.
\medskip

\noin
\end{abstract}

\section{Motivation}

Formal definitions are postponed until Section~\ref{sec:basics}.

\smallskip

The first comprehensive  paper on state complexity was published by A. N. Maslov~\cite{Mas70} in 1970, but this work was unknown in the West for many years. Maslov wrote:  
 \begin{quote} {\it An important measure of the complexity of [sets of words representable in finite automata] is the number of states in the minimal representing automaton.
... if $T(A) \cup T(B)$ are representable in automata $A$ and $B$ with $m$ and $n$ states respectively ..., then:
	\be 
	\item
	$T(A) \cup T(B)$ is representable in an automaton with $m\cdot n$ states;
	\item
	$T(A).T(B)$ is representable in an automaton with $(m-1)2^n + 2^{n-1}$ states.
	\ee}
\end{quote}
The second comprehensive paper on state complexity was published by S. Yu, Q. Zhuang and K. Salomaa~\cite{YZS94} in 1994. Here the authors wrote:
\begin{quote}{\it
\be
\item
... for any pair of complete $m$-state DFA $A$ and $n$-state DFA $B$ defined on the same alphabet $\Sigma$, there exists a DFA with at most $m2^n-2^{n-1}$ states which accepts $L(A)L(B)$.
\item
... $m\cdot n$ states are ... sufficient for a DFA to accept the intersection (union) of an $m$-state DFA language and an $n$-state DFA language.
\ee}
\end{quote}

Here DFA stands for \emph{deterministic finite automaton}, and \emph{complete} means that there is a transition from every state under every input letter.

I will show that statements 1 and 2 of Maslov are incorrect, but undoubtedly Maslov had in mind languages over the same alphabet, in which case the statements are correct. 
In~\cite{YZS94} the first statement includes the same-alphabet restriction, but the second omits it (presumably it is implied by the context). 

The same-alphabet restriction is unnecessary: There is no reason why we should not be able to find, for example, the union of 
languages $L'_2=\{a,b\}^*b$ and $L_2=\{a,c\}^*c$  accepted by the minimal complete two-state automata $\cD'_2$ and $\cD_2$ of Figure~\ref{fig:example}, where an incoming arrow denotes the initial state and a double circle represents a final state.

\begin{figure}[ht]
\unitlength 8.5pt
\begin{center}\begin{picture}(30,4)(0,6)
\gasset{Nh=2,Nw=2,Nmr=1.25,ELdist=0.4,loopdiam=1.5}
\node(0')(1,7){$0'$}\imark(0')
\node(1')(8,7){$1'$}\rmark(1')
\node(0)(22,7){0}\imark(0)
\node(1)(29,7){1}\rmark(1)
\drawloop(0'){$a$}
\drawloop(1'){$b$}
\drawedge[curvedepth= .8,ELdist=.4](0',1'){$b$}
\drawedge[curvedepth= .8,ELdist=.4](1',0'){$a$}
\drawloop(0){$a$}
\drawloop(1){$c$}
\drawedge[curvedepth= .8,ELdist=.4](0,1){$c$}
\drawedge[curvedepth= .8,ELdist=.4](1,0){$a$}
\end{picture}\end{center}
\caption{Two minimal complete DFAs $\cD'_2$ and $\cD_2$.}
\label{fig:example}
\end{figure}
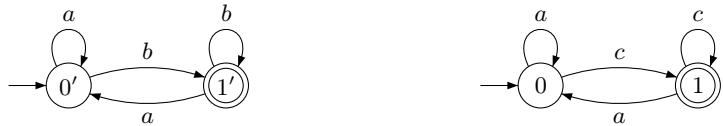
The union of $L'_2$ and $L_2$ is a language over three letters. To find the DFA for $L'_2 \cup L_2$, we view $\cD'_2$ and $\cD_2$ as incomplete DFA's, the first missing all transitions under $c$, and the second under $b$. 
After adding the missing transitions we obtain DFAs $\cD'_3$ and $\cD_3$ shown in Figure~\ref{fig:complete}. Now we can proceed as is usually done in the same-alphabet approach, and take the direct product of 
$\cD'_3$ and $\cD_3$ to find $L_2'\cup L_2$. Here it turns out that six states are necessary to represent $L'_2\cup L_2$, but the state complexity of  union is actually 
$(m+1)(n+1)$. 

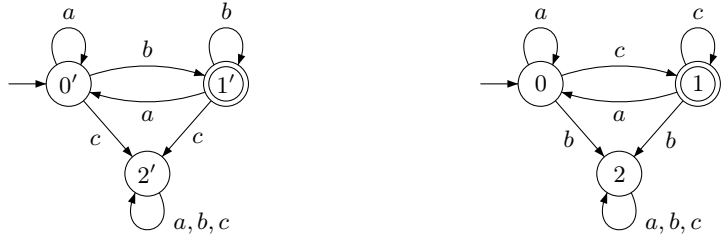
\begin{figure}[ht]
\unitlength 8.5pt
\begin{center}\begin{picture}(30,9)(0,1.6)
\gasset{Nh=2,Nw=2,Nmr=1.25,ELdist=0.4,loopdiam=1.5}
\node(0')(1,7){$0'$}\imark(0')
\node(1')(8,7){$1'$}\rmark(1')
\node(2')(4.5,3){$2'$}
\node(0)(22,7){0}\imark(0)
\node(1)(29,7){1}\rmark(1)
\node(2)(25.5,3){2}
\drawloop(0'){$a$}
\drawloop(1'){$b$}
\drawedge[curvedepth= .8,ELdist=.4](0',1'){$b$}
\drawedge[curvedepth= .8,ELdist=.4](1',0'){$a$}
\drawloop[loopangle=270,ELpos=25](2'){$a,b,c$}
\drawloop(0){$a$}
\drawloop(1){$c$}
\drawedge[ELdist=-1.1](0',2'){$c$}
\drawedge[ELdist=.3](1',2'){$c$}
\drawedge[curvedepth= .8,ELdist=.4](0,1){$c$}
\drawedge[curvedepth= .8,ELdist=.4](1,0){$a$}
\drawedge[ELdist=-1.1](0,2){$b$}
\drawedge[ELdist=.3](1,2){$b$}
\drawloop[loopangle=270,ELpos=25](2){$a,b,c$}
\end{picture}\end{center}
\caption{DFAs $\cD'_3$ and $\cD_3$ over three letters.}
\label{fig:complete}
\end{figure}

In general, when calculating the result of a binary operation on regular languages with different alphabets, we deal with special incomplete DFAs that are only missing some letters and all the transitions caused by these letters. The complexity of incomplete DFAs has been studied previously by Gao, K. Salomaa, and Yu~\cite{GSY11} and by Maia, Moreira and Reis~\cite{MMR15}. However, the objects studied there are \emph{arbitrary} incomplete DFAs, whereas we are interested only in \emph{complete DFAs with some missing letters}. 
Secondly, we study \emph{state} complexity, whereas the above-mentioned papers deal mainly with \emph{transition} complexity.
Nevertheless, there is some overlap. It was shown in~\cite[Corollary 3.2]{GSY11} that the incomplete state complexity of union is less than or equal to $mn+m+n$, and that this bound is tight in some special cases. In~\cite[Theorem 2]{MMR15}, witnesses that work in all cases were found. These complexities correspond to my result for union in Theorem~\ref{thm:boolean}.
Also in~\cite[Theorem 5]{MMR15}, the incomplete state complexity of product is shown to be $m2^n+2^{n-1}-1$, and this corresponds to my result for product in Theorem~\ref{thm:product}.

In this paper I remove the restriction of equal alphabets of the two operands. 
I prove that
the complexity of union and symmetric difference is $mn+m+n+1$, that of intersection is $mn$, that of difference is $mn+m$,                                                                                                     
and that of the product is  $m2^n+2^{n-1}$, if each language's own alphabet is used. 
I exhibit a new most complex regular language that meets the complexity bounds for boolean operations, product, star, and reversal, has a maximal syntactic semigroup and most complex atoms.  All the witnesses used here are derived from that one most complex language.

\section{Terminology and Notation}
\label{sec:basics}

We say that the \emph{alphabet of a regular language} $L$ is $\Sig$ (or that \emph{$L$ is a language over $\Sig$}) if $L\subseteq \Sig^*$ and for each letter $a\in \Sig$ there is a word $uav$ in $L$.
A basic complexity measure of $L$ with alphabet $\Sig$ is the number $n$ of distinct (left) quotients of $L$ by words in $\Sig^*$, where a \emph{(left) quotient} of $L$ by a word $w\in\Sig^*$ is $w^{-1}L=\{x\mid wx\in L\}$. 
The number of quotients of $L$ is its \emph{quotient complexity}~\cite{Brz10a}, $\kappa(L)$. 
A concept equivalent to quotient complexity is the \emph{state complexity}~\cite{YZS94} of $L$, which is the number of states in a complete minimal deterministic finite automaton (DFA) recognizing $L$.
Since we do not use any other measures of complexity in this paper (with the exception of one mention of time and space complexity in the next paragraph), we refer to quotient/state complexity simply as \emph{complexity}.

Let $L'_m\subseteq\Sig'^*$ and $L_n\subseteq \Sig^*$  be regular languages of  complexities $m$ and $n$, respectively.
The \emph{complexity of a binary operation}  $\circ$ on $L'_m$ and $L_n$ 
is  the maximal value of $\kappa(L'_m \circ L_n)$ as a function $f(m,n)$,
as $L'_m$ and $L_n$ range over all regular languages of complexity $m$ and $n$, respectively.
The complexity of an operation gives a worst-case lower bound on the time and space complexity of the operation. For this reason it has been studied extensively; see~\cite{Brz10a,Brz13,Yu01,YZS94} for additional references.

A \emph{deterministic finite automaton (DFA)} is a quintuple
$\cD=(Q, \Sigma, \delta, q_0,F)$, where
$Q$ is a finite non-empty set of \emph{states},
$\Sig$ is a finite non-empty \emph{alphabet},
$\delta\colon Q\times \Sig\to Q$ is the \emph{transition function},
$q_0\in Q$ is the \emph{initial} state, and
$F\subseteq Q$ is the set of \emph{final} states.
We extend $\delta$ to a function $\delta\colon Q\times \Sig^*\to Q$ as usual.
A~DFA $\cD$ \emph{accepts} a word $w \in \Sigma^*$ if ${\delta}(q_0,w)\in F$. The language accepted by $\cD$ is denoted by $L(\cD)$. If $q$ is a state of $\cD$, then the language $L^q$ of $q$ is the language accepted by the DFA $(Q,\Sigma,\delta,q,F)$. 
A state is \emph{empty} (or \emph{dead} or a \emph{sink state}) if its language is empty. Two states $p$ and $q$ of $\cD$ are \emph{equivalent} if $L^p = L^q$. 
A state $q$ is \emph{reachable} if there exists $w\in\Sig^*$ such that $\delta(q_0,w)=q$.
A DFA is \emph{minimal} if all of its states are reachable and no two states are equivalent.
Usually DFAs are used to establish upper bounds on the complexity of operations, and also as witnesses that meet these bounds.

If $\delta(q,a)=p$ for a state  $q\in Q$ and a letter $a\in \Sig$, we say there is a \emph{transition} under $a$ from $q$ to $p$ in $\cD$.
The DFAs defined above are \emph{complete} in the sense that there is \emph{exactly one} transition for each state  $q\in Q$ and each letter $a\in \Sig$. If there is \emph{at most one transition}  for each state of $Q$ and letter of $\Sig$,  the automaton is an \emph{incomplete} DFA.

A \emph{nondeterministic finite automaton (NFA)} is a 5-tuple
$\cD=(Q, \Sigma, \delta, I,F)$, where
$Q$,
$\Sig$ and $F$ are defined as in a DFA, 
$\delta\colon Q\times \Sig\to 2^Q$ is the \emph{transition function}, and
$I\subseteq Q$ is the \emph{set of initial states}. 
An \emph{$\eps$-NFA} is an NFA in which transitions under the empty word $\eps$ are also permitted.

To simplify the notation, without loss of generality we use $Q_n=\{0,\dots,n-1\}$ as the set of states of every DFA with $n$ states.
A \emph{transformation} of $Q_n$ is a mapping $t\colon Q_n\to Q_n$.
The \emph{image} of $q\in Q_n$ under $t$ is denoted by $qt$.
For $k\ge 2$, a transformation (permutation) $t$ of a set $P=\{q_0,q_1,\ldots,q_{k-1}\} \subseteq Q$ is a \emph{$k$-cycle}
if $q_0t=q_1, q_1t=q_2,\ldots,q_{k-2}t=q_{k-1},q_{k-1}t=q_0$.
This $k$-cycle is denoted by $(q_0,q_1,\ldots,q_{k-1})$, and acts as the identity on the states in $Q_n\setminus P$.
A~2-cycle $(q_0,q_1)$ is called a \emph{transposition}.
 A transformation that changes only one state $p$ to a state $q\neq p$ and acts as the identity for the other states is denoted by $(p\to q)$.
 The identity transformation is denoted by $\one$.

In any DFA, each $a\in \Sig$ induces a transformation $\delta_a$ of the set $Q_n$ defined by $q\delta_a=\delta(q,a)$; we denote this by $a\colon \delta_a$. 
For example, when defining the transition function of a DFA, we write $a\colon (0,1)$ to mean that
$\delta(q,a)=q(0,1)$, where the transformation $(0,1)$ acts on state $q$ as follows: if $q$ is 0  it maps it to 1, if $q$ is 1 it maps it to 0,
and it acts as the identity on the remaining states.

By a slight abuse of notation we use the letter $a$ to denote the transformation it induces; thus we write $qa$ instead of $q\delta_a$.
We extend the notation to sets of states: if $P\subseteq Q_n$, then $Pa=\{pa\mid p\in P\}$.
We also find it convenient to write $P\stackrel{a}{\longrightarrow} Pa$ to indicate that the image of $P$ under $a$ is $Pa$.
If $s,t$ are transformations of $Q$, their composition is denoted by $s\ast t$ and defined by
$q(s \ast t)=(qs)t$; the $\ast$ is usually omitted.
Let $\cT_{Q_n}$ be the set of all $n^n$ transformations of $Q_n$; then $\cT_{Q_n}$ is a monoid under composition.

A sequence $(L_n, n\ge k)=(L_k,L_{k+1},\dots)$, of regular languages is called a \emph{stream}; here $k$ is usually some small integer, and the languages in the stream usually have the same form and differ only in the parameter $n$. For example, $(\{a,b\}^*a^n\{a,b\}^* \mid n\ge 2)$ is a stream.
To find the complexity of a binary operation $\circ$ we need to find an upper bound on this complexity and two streams 
 $(L'_m, m \ge h)$ and $(L_n, n\ge k)$ of languages meeting this bound.
 In general, the two streams are different, but there are many examples where $L'_n$ ``differs only slightly'' from $L_n$; such a language $L'_n$ is called a \emph{dialect}~\cite{Brz13} of $L_n$. 

Let $\Sig=\{a_1,\dots,a_k\}$ be an alphabet; we assume that its elements are ordered as shown.
Let $\pi$ be a \emph{partial permutation} of $\Sig$, that is, a partial function $\pi \colon \Sig \rightarrow \Gamma$ where $\Gamma \subseteq \Sig$, for which there exists $\Delta \subseteq \Sig$ such that $\pi$ is bijective when restricted to $\Delta$ and  undefined on $\Sig \setminus \Delta$. 
We denote undefined values of $\pi$ by  ``$-$'', that is, we write $\pi(a)=-$, if $\pi$ is undefined at $a$.

If $L\subseteq \Sig^*$, we denote it by $L(a_1,\dots,a_k)$ to stress its dependence on $\Sig$.
If $\pi$ is a partial permutation, let $s_\pi(L(a_1,\dots,a_k))$ be the language obtained from $L(a_1,\dots,a_k)$ by the substitution  $s_\pi$ defined as follows: 
 for $a\in \Sig$, 
$a \mapsto \{\pi(a)\}$ if $\pi(a)$ is defined, and $a \mapsto \emp$ otherwise.
The \emph{permutational dialect}, or simply \emph{dialect},  of $L(a_1,\dots,a_k)$ defined by $\pi$ is the  language 
$L(\pi(a_1),\dotsc,\pi(a_k))= s_\pi(L(a_1,\dots,a_k))$.

Similarly,
let $\cD = (Q,\Sig,\delta,q_0,F)$ be a DFA; we denote it by
$\cD(a_1,\dots,a_k)$ to stress its dependence on $\Sig$.
If $\pi$ is a partial permutation, then the \emph{permutational dialect}, or simply \emph{dialect}, 
$\cD(\pi(a_1),\dotsc,\pi(a_k))$ of
$\cD(a_1,\dots,a_k)$ is obtained by changing the alphabet of $\cD$ from $\Sig$ to $\pi(\Sig)$, and modifying $\delta$ so that in the modified DFA 
$\pi(a_i)$ induces the transformation induced by $a_i$  in the original DFA.
One verifies that if the language $L(a_1,\dots,a_k)$ is accepted by DFA $\cD(a_1,\dots,a_k)$, then
$L(\pi(a_1),\dotsc,\pi(a_k))$ is accepted by $\cD(\pi(a_1),\dotsc,\pi(a_k))$.

If the letters for which $\pi$ is undefined are at the end of the alphabet $\Sig$, then they are omitted. For example,
if $\Sig=\{a,b,c,d\}$ and $\pi(a)=b$, $\pi(b)=a$, and $\pi(c)=\pi(d)=-$, then we write $L_n(b,a)$ for $L_n(b,a,-,-)$, etc.
     \medskip

\section{Boolean Operations}

A binary boolean operation is \emph{proper} if it is not a constant and does not depend on only one variable.
We  study the complexities  of four proper boolean operations only: union ($\cup$), symmetric difference ($\oplus$),
difference ($\setminus$), and intersection ($\cap$); the  complexity of any other proper operation can be deduced from these four.
For example, $\kappa(\ol{L'} \cup L)=\kappa\left( \ol{ \ol{L'}\cup L}\right)=\kappa (L'\cap \ol{L})=\kappa(L'\setminus L)$, where we have used the well-known fact that $\kappa(\ol{L})=\kappa(L)$, for any $L$.

The DFA of Definition~\ref{def:regular} is required for the next theorem; this DFA is the 4-input ``universal witness'' called $\cU_n(a,b,c,d)$ in~\cite{Brz13}.
\begin{definition}
\label{def:regular}
For $n\ge 3$, let $\cD_n=\cD_n(a,b,c,d)=(Q_n,\Sig,\delta_n, 0, \{n-1\})$, where 
$\Sig=\{a,b,c,d\}$, 
and $\delta_n$ is defined by the transformations
$a\colon (0,\dots,n-1)$,
$b\colon(0,1)$,
$c\colon(n-1 \rightarrow 0)$, and
$d\colon \one$.
Let $L_n=L_n(a,b,c,d)$ be the language accepted by~$\cD_n$.
The structure of $\cD_n(a,b,c,d)$ is shown in Figure~\ref{fig:RegWit}. 
\end{definition}

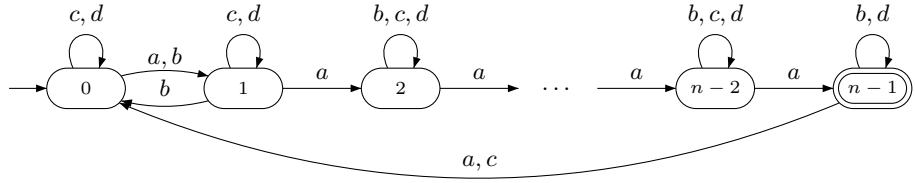
\begin{figure}[ht]
\unitlength 8.5pt
\begin{center}\begin{picture}(37,10)(0,2)
\gasset{Nh=1.8,Nw=3.5,Nmr=1.25,ELdist=0.4,loopdiam=1.5}
	{\scriptsize
\node(0)(1,7){0}\imark(0)
\node(1)(8,7){1}
\node(2)(15,7){2}
}
\node[Nframe=n](3dots)(22,7){$\dots$}
	{\scriptsize
\node(n-2)(29,7){$n-2$}
	}
	{\scriptsize
\node(n-1)(36,7){$n-1$}\rmark(n-1)
	}
\drawloop(0){$c,d$}
\drawedge[curvedepth= .8,ELdist=.1](0,1){$a,b$}
\drawedge[curvedepth= .8,ELdist=-1.2](1,0){$b$}
\drawedge(1,2){$a$}
\drawloop(2){$b,c,d$}
\drawedge(2,3dots){$a$}
\drawedge(3dots,n-2){$a$}
\drawloop(n-2){$b,c,d$}
\drawedge(n-2,n-1){$a$}
\drawedge[curvedepth= 4.0,ELdist=-1.0](n-1,0){$a,c$}
\drawloop(n-1){$b,d$}
\drawloop(1){$c,d$}
\end{picture}\end{center}
\caption{ DFA  of  Definition~\ref{def:regular}.}
\label{fig:RegWit}
\end{figure}

\begin{theorem}
\label{thm:boolean}
For $m,n \ge 3$, let $L'_m$ (respectively, $L_n$) be a regular language with $m$ (respectively, $n$) quotients over an alphabet $\Sig'$, (respectively, $\Sig$). 
Then the complexity of union and symmetric difference is $mn+m+n+1$ 
and this bound is met by $L'_m(a,b,-,c)$ and $L_n(b,a,-,d)$; 
the complexity of difference is $mn+m$, and this bound is met by $L'_m(a,b,-,c)$ and
$L_n(b,a)$; the complexity of intersection is $mn$ and this bound is met by  
$L'_m(a,b)$ and $L_n(b.a)$.

\end{theorem}

\begin{proof}

Let $\cD'_m = ( Q'_m, \Sig', \delta', 0',F')$  and 
$\cD_n = (Q_n, \Sig,\delta, 0, F)$ be minimal DFAs for $L'_m$ and $L_n$, respectively.
To calculate an upper bound for the boolean operations assume that $\Sig'\setminus \Sig $ and 
$\Sig \setminus \Sig' $ are non-empty; this assumption results in the largest upper bound. 
We  add an empty state $\emp'$ to $\cD'_m$ to send all transitions under the letters from 
$\Sig \setminus \Sig' $ to that state; thus we get an $(m+1)$-state DFA  $\cD'_{m,\emp'}$. Similarly, we add an empty state $\emp$  to $\cD_n$ to get $\cD_{n,\emp}$.
Now we have two DFAs over the same alphabet, and an ordinary problem of finding an upper bound for the boolean operations on two languages over the same alphabet, \emph{except that these languages both have empty quotients}. 
It is clear that $(m+1)(n+1)$ is an upper bound for all four operations, but it can be improved for difference and intersection.
Consider the direct product $\cP_{m,n}$ of $\cD'_{m,\emp'}$ and $\cD_{n,\emp}$. 

For difference, all $n+1$ states of $\cP_{m,n}$ that have the form $(\emp', q)$, where $q\in Q_n\cup \{\emp\}$ are empty. Hence the bound can be reduced by $n$ states to $mn+m+1$. 
However, the empty states can only be reached by words in $\Sig\setminus \Sig'$ and the alphabet of $L'_m\setminus L_n$ is a subset of $\Sig'$; hence  the bound is reduced futher to $mn+m$.

For intersection, all $n$ states $(\emp',q)$, $q\in Q_n$, and all $m$ states $(p',\emp)$, $p'\in Q'_m$, are equivalent to the empty state $(\emp',\emp)$, thus reducing the upper bound to
$mn+1$. Since the alphabet of $L'_m\cap L_n$ is a subset of $\Sig'\cap \Sig$, these empty states cannot be reached and the bound is reduced to $mn$.

To prove that the bounds are tight, we start with $\cD_n(a,b,c,d)$ of Definition~\ref{def:regular}.
For $m,n\ge 3$,  let $D'_m(a,b,-,c)$ be the dialect of $\cD'_m(a,b,c,d)$ where $c$ plays the role of $d$ and the alphabet is restricted to $\{a,b,c\}$, and let
$\cD_n(b,a,-,d)$ be the dialect of $\cD_n(a,b,c,d)$ in which $a$ and $b$ are permuted, and the alphabet is  restricted to $\{a,b,d\}$;
see Figure~\ref{fig:boolean}.

\begin{figure}[ht]
\unitlength 7.5pt
\begin{center}\begin{picture}(37,17)(-3.5,2)
\gasset{Nh=2.2,Nw=5.0,Nmr=1.25,ELdist=0.4,loopdiam=1.5}
	{\scriptsize
\node(0')(-2,14){$0'$}\imark(0')
\node(1')(7,14){$1'$}
\node(2')(16,14){$2'$}
\node[Nframe=n](3dots')(25,14){$\dots$}
\node(m-1')(34,14){$(m-1)'$}\rmark(m-1')
}
\drawedge[curvedepth= 1.4,ELdist=-1.3](0',1'){$a,b$}
\drawedge[curvedepth= 1,ELdist=.3](1',0'){$b$}
\drawedge(1',2'){$a$}
\drawedge(2',3dots'){$a$}
\drawedge(3dots',m-1'){$a$}
\drawedge[curvedepth= -5.2,ELdist=-1](m-1',0'){$a$}
\drawloop(0'){$c$}
\drawloop(1'){$c$}
\drawloop(2'){$b,c$}
\drawloop(m-1'){$b,c$}

\gasset{Nh=2.2,Nw=5.0,Nmr=1.25,ELdist=0.4,loopdiam=1.5}
	{\scriptsize
\node(0)(-2,7){0}\imark(0)
\node(1)(7,7){1}
\node(2)(16,7){2}
\node[Nframe=n](3dots)(25,7){$\dots$}
\node(n-1)(34,7){$n-1$}\rmark(n-1)
	}
\drawloop(0){$d$}
\drawloop(1){$d$}
\drawloop(2){$a,d$}
\drawloop(n-1){$a,d$}
\drawedge[curvedepth= 1.2,ELdist=-1.3](0,1){$a,b$}
\drawedge[curvedepth= .8,ELdist=.25](1,0){$a$}
\drawedge(1,2){$b$}
\drawedge(2,3dots){$b$}
\drawedge(3dots,n-1){$b$}
\drawedge[curvedepth= 5.0,ELdist=-1.5](n-1,0){$b$}

\end{picture}\end{center}
\caption{Witnesses $D'_m(a,b,-,c)$ and $\cD_n(b,a,-,d)$ for boolean operations. }
\label{fig:boolean}
\end{figure}
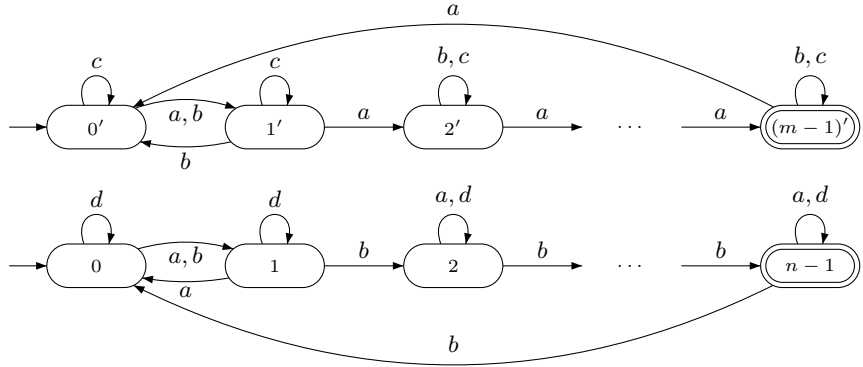

To finish the proof, we complete the two DFAs by adding empty states, and construct their direct product as illustrated in Figure~\ref{fig:cross}.
If we restrict both DFAs to the alphabet $\{a,b\}$,  we have the usual problem of determining the complexity of two DFAs over the same alphabet. 
By \cite[Theorem 1]{BBMR14}, all $mn$ states of the form $(p',q)$, $p'\in Q'_m$, $q\in Q_n$, are reachable and pairwise distinguishable by words in $\{a,b\}^*$ for all proper boolean operations if $(m,n)\notin \{(3,4),(4,3),(4,4)\}$.
For our application, the three exceptional cases were verified by computation.

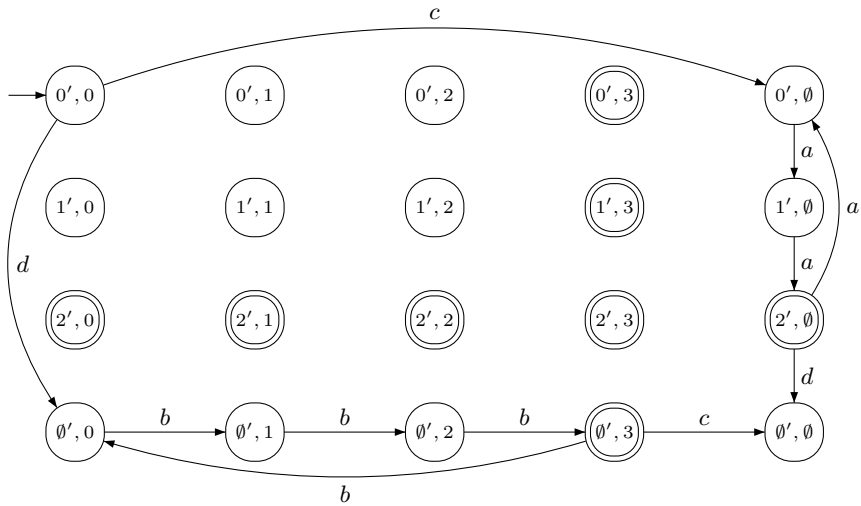
\begin{figure}[ht]
\unitlength 8.5pt
\begin{center}\begin{picture}(35,20)(0,-2)
\gasset{Nh=2.6,Nw=2.6,Nmr=1.2,ELdist=0.3,loopdiam=1.2}
	{\scriptsize
\node(0'0)(2,15){$0',0$}\imark(0'0)
\node(1'0)(2,10){$1',0$}
\node(2'0)(2,5){$2',0$}\rmark(2'0)
\node(3'0)(2,0){$\emp',0$}

\node(0'1)(10,15){$0',1$}
\node(1'1)(10,10){$1',1$}
\node(2'1)(10,5){$2',1$}\rmark(2'1)
\node(3'1)(10,0){$\emp',1$}

\node(0'2)(18,15){$0',2$}
\node(1'2)(18,10){$1',2$}
\node(2'2)(18,5){$2',2$}\rmark(2'2)
\node(3'2)(18,0){$\emp',2$}

\node(0'3)(26,15){$0',3$}\rmark(0'3)
\node(1'3)(26,10){$1',3$}\rmark(1'3)
\node(2'3)(26,5){$2',3$}\rmark(2'3)
\node(3'3)(26,0){$\emp',3$}\rmark(3'3)

\node(0'4)(34,15){$0',\emp$}
\node(1'4)(34,10){$1',\emp$}
\node(2'4)(34,5){$2',\emp$}\rmark(2'4)
\node(3'4)(34,0){$\emp',\emp$}
	}
\drawedge(3'0,3'1){$b$}
\drawedge(3'1,3'2){$b$}
\drawedge(3'2,3'3){$b$}
\drawedge[curvedepth=2,ELdist=.4](3'3,3'0){$b$}

\drawedge(0'4,1'4){$a$}
\drawedge(1'4,2'4){$a$}
\drawedge[curvedepth=-2,ELdist=-.9](2'4,0'4){$a$}
\drawedge(3'3,3'4){$c$}
\drawedge(2'4,3'4){$d$}

\drawedge[curvedepth=-3,ELdist=.4](0'0,3'0){$d$}
\drawedge[curvedepth=3,ELdist=.4](0'0,0'4){$c$}

\end{picture}\end{center}
\caption{Direct product for union shown partially.}
\label{fig:cross}
\end{figure}

To prove that the remaining states are reachable, observe that 
$(0',0) \stackrel{d} {\longrightarrow} (\emp',0)$ and
$(\emp',0) \stackrel{b^q} {\longrightarrow} (\emp',q)$, for $q\in Q_n$.
Symmetrically, 
$(0',0) \stackrel{c} {\longrightarrow} (0',\emp)$ and
$(0',\emp) \stackrel{a^p} {\longrightarrow} (p',\emp)$, for $p'\in Q'_m$.
Finally, $(\emp',n-1) \stackrel{c} {\longrightarrow} (\emp',\emp)$, and  all $(m+1)(n+1)$ states of the direct product are reachable.

It remains to verify that the appropriate states are pairwise distinguishable.
From \cite[Theorem 1]{BBMR14}, we know that all states in $Q'_m\times Q_n$ are distinguishable.
Let $H= \{(\emp',q) \mid q\in Q_n  \}$, and
$V= \{ (p',\emp) \mid p'\in Q'_m \}$.
For the operations consider four cases:
\bd

\item[Union]
The final states of $\cP_{m,n}$ are $\{((m-1)',q) \mid q\in Q_n\cup \{\emp\}  \}$,  and
$\{ (p',n-1) \mid p'\in Q'_m \cup \{\emp'\} \}$.
Every state in $V$ accepts a word with a $c$, whereas  no state in $H$ accepts such words.
Similarly, every state in $H$ accepts a word with a $d$, whereas  no state in $V$ accepts such words.
Every state in $Q'_m \times Q_n$ accepts a word with a $c$ and a word with a $d$. State $(\emp',\emp)$ accepts no words at all.
Hence any two states chosen from different sets (the sets being $Q'_m\times Q_n$, $H$, $V$,  and $\{(\emp',\emp)\}$) are distinguishable.
States in $H$ are distinguishable by words in $b^*$ and those in $V$, by words in $a^*$.
Therefore all $mn+m+n+1$ states are pairwise distinguishable.

\item[Symmetric Difference]
The final states here are all the final states for union except $( (m-1)',n-1 )$. The rest of the argument is the same as for union.

\item[Difference]
Here the final states are $\{((m-1)', q) \mid q\neq n-1\}$.
The $n$ states of the form $(\emp',q)$, $q \in Q_n$, are now equivalent to the empty state $(\emp',\emp)$. 
The remaining states are pairwise distinguishable by the arguments used for union. Hence we have $mn+m+1$ distinguishable states.
However, the alphabet of $L'_m \setminus L_n$ is $\{a,b,c\}$, and  the empty state can only be reached by $d$. Since this empty state is not needed, neither is $d$,  the final bound is $mn+m$  and it is reached by $L'_m(a,b,-,c)$ and $L_n(b,a)$.

\item[Intersection]
Here only $((m-1)', n-1 )$ is final and all states $(p', \emp)$, $p' \in Q'_m$, and $(\emp',q)$, $q\in Q_n$ are equivalent to $(\emp',\emp)$, leaving $mn+1$ distinguishable states.
 However, the alphabet of $L'_m \cap L_n$ is $\{a,b\}$, and so the empty state cannot be reached. This gives the final bound of $mn$ states, and this bound is met by $L'_m(a,b)$ and $L_n(b.a)$ as was already known in~\cite{Brz13}.
\qed
\ed
\end{proof}

\section{Product}

\begin{theorem}
\label{thm:product}
For $m,n \ge 3$, let $L'_m$ (respectively, $L_n$) be a regular language with $m$ (respectively, $n$) quotients over an alphabet $\Sig'$, (respectively, $\Sig$). 
Then $\kappa(L'_m L_n) \le m2^n+2^{n-1}$, and this bound is met by $L'_m(a,b,-,c)$ and $L_n(b,a,-,d)$.
\end{theorem}
\begin{proof}
First we derive the upper bound.
Let $\cD'_m=( Q'_m, \Sig', \delta', 0',F')$  and $\cD_n=(Q_n,\Sig,\delta,0,F) $ be minimal DFAs of $L'_m$  and $L_n$, respectively.
We use the normal construction of an $\eps$-NFA $\cN$ to recognize $L'_mL_n$, by introducing an $\eps$-transition from each final state of $\cD'_m$ to the initial state of $\cD_n$, and changing all final states of $\cD'_m$ to non-final. This is illustrated in Figure~\ref{fig:product}, where $(m-1)'$ is the only final state of $\cD'_m$.
We then determinize $\cN$ using the subset construction to get the DFA $\cD$ for $L'_m L_n$.

Suppose $\cD'_m$ has $k$ final states, where $1 \le k \le m-1$.
I will show that $\cD$ can have only the following types of states: (a)  
at most $(m-k)2^n$ states $\{p'\} \cup S$, where $p'\in Q'_m\setminus F'$, and $S\subseteq Q_n$, 
(b) at most $k2^{n-1}$ states  
$\{ p' , 0\} \cup S$, where $p'\in F'$ and $S\subseteq Q_n\setminus \{0\}$,
and (c) at most $2^n$ states $S\subseteq Q_n$.
Because $\cD'_m$ is deterministic, there can be at most one state $p'$ of $Q'_m$ in any reachable subset. 
If $p' \notin F'$, it may be possible to reach any subset of states of $Q_n$ along with $p'$, and this accounts for (a).
If $p' \in F'$, then the set must contain $0$ and possibly any subset of $Q_n\setminus \{0\}$, giving (b).
It may also be possible to have any subset $S$ of $Q_n$ by applying an input that is not in $\Sig'$ to $\{0'\} \cup S$ to get $S$, and so we have~(c).
Altogether, there are at most $(m-k)2^n +k2^{n-1}+2^n = (2m-k)2^{n-1} + 2^n$ reachable subsets. This expression reaches its maximum when $k=1$, and hence we have at most $m2^n+2^{n-1}$ states in $\cD$.

\begin{figure}[ht]
\unitlength 7.5pt
\begin{center}\begin{picture}(37,17)(-4,2)
\gasset{Nh=2.0,Nw=4.6,Nmr=1.25,ELdist=0.4,loopdiam=1.5}
	{\scriptsize
\node(0')(-4,14){$0'$}\imark(0')
\node(1')(3,14){$1'$}
\node(2')(10,14){$2'$}
\node[Nframe=n](3dots')(17,14){$\dots$}
\node(m-1')(24,14){$(m-1)'$}
	}
\drawedge[curvedepth= 1.4,ELdist=-1.3](0',1'){$a,b$}
\drawedge[curvedepth= 1,ELdist=.3](1',0'){$b$}
\drawedge(1',2'){$a$}
\drawedge(2',3dots'){$a$}
\drawedge(3dots',m-1'){$a$}
\drawedge[curvedepth= -5.2,ELdist=-1](m-1',0'){$a$}
\drawloop(0'){$c$}
\drawloop(1'){$c$}
\drawloop(2'){$b,c$}
\drawloop(m-1'){$b,c$}

\gasset{Nh=2.0,Nw=3.7,Nmr=1.25,ELdist=0.4,loopdiam=1.5}
	{\scriptsize
\node(0)(7,7){0}\imark(0)
\node(1)(14,7){1}
\node(2)(21,7){2}
\node[Nframe=n](3dots)(28,7){$\dots$}
\node(n-1)(35,7){$n-1$}\rmark(n-1)
	}
\drawloop(0){$d$}
\drawloop(1){$d$}
\drawloop(2){$a,d$}
\drawloop(n-1){$a,d$}
\drawedge[curvedepth= 1.2,ELdist=-1.3](0,1){$a,b$}
\drawedge[curvedepth= .8,ELdist=.25](1,0){$a$}
\drawedge(1,2){$b$}
\drawedge(2,3dots){$b$}
\drawedge(3dots,n-1){$b$}
\drawedge[curvedepth= 5.0,ELdist=-1.5](n-1,0){$b$}
\drawedge[curvedepth= -1.5,ELdist=-1](m-1',0){$\eps$}
\end{picture}\end{center}
\caption{An NFA  for the product of  $L'_m(a,b,-,c)$ and $L_n(b,a,-,d)$. }
\label{fig:product}
\end{figure}
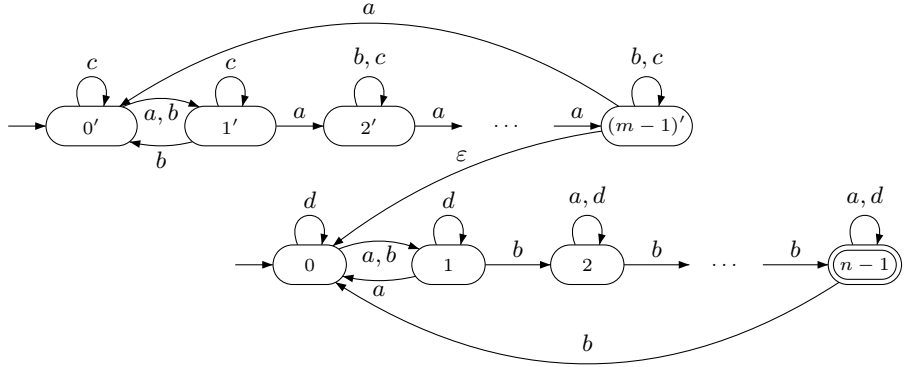

To prove that the bound is tight, we use the same witnesses as for boolean operations; see Figure~\ref{fig:product}. 
If $S=\{q_1,\dots, q_k\}\subseteq Q_n$ then $S+i= \{q_1+i,\dots, q_k+i\}$ and $S-i = \{q_1-i,\dots, q_k-i\} $,
where  addition and subtraction are modulo $n$.
Note that $b^2$ and $a^m$ ($a^2$ and $b^n$) act as the identity on $Q'_m$ ($Q_n$). 
If $p < m-1$, then $\{p'\} \cup S \stackrel{b^2}{\longrightarrow} \{p'\} \cup (S+2)$, for all $S\subseteq  Q_n$.
If $n$ is odd, then $(b^2)^{(n-1)/2} =b^{n-1}$ and 
$\{p'\} \cup S \stackrel{ b^{n-1} }{\longrightarrow} \{ p' \} \cup (S-1)$, for all $q \in Q_n$.
If $0, 1\notin S$ or $\{0,1\}\subseteq S$, then $a$ acts as the identity on $S$.

\begin{remark}
\label{rem:add1}
If $1\notin S$ and $\{(m-2)'\} \cup S$ is reachable, then 
$\{0',1\} \cup S$ is reachable for all $S\subseteq  Q_n\setminus \{1\}$.
\end{remark}
\begin{proof}
If $0\in S$, then 
$\{(m-2)',0\} \cup S\setminus\{0\} \stackrel{a}{\longrightarrow} \{(m-1)',0,1\} \cup S\setminus\{0\} 
		\stackrel{a}{\longrightarrow} \{ 0',0,1 \} \cup S\setminus\{0\} = \{0',1\} \cup S. $	
If $0\notin S$,	then	
$\{(m-2)'\} \cup S \stackrel{a}{\longrightarrow} \{(m-1)',0\} \cup S 
		\stackrel{a}{\longrightarrow} \{ 0',1 \} \cup S. $
\qed	
\end{proof}

We now prove that the languages of Figure~\ref{fig:product} meet the upper bound.
\smallskip

\noin 
{\bf Claim 1:}
\emph{All sets of the form $\{p'\} \cup S$, where $p'\in Q'_{m-1}$ and $S\subseteq Q_n$, are reachable.}
We show this by induction on the size of $S$.
\smallskip

\noin
{\bf Basis: $|S| =0$.}
The initial set is $\{0'\}$, and from $\{0'\}$ we reach  $\{p'\}$, $p'\in Q'_{m-1}$,  by  $a^{p}$, without reaching any states of $Q_n$. Thus the claim holds if $|S|=0$.
\smallskip

\noin
{\bf Induction Assumption:} $\{p'\} \cup S$, where $p'\in Q'_{m-1}$ and $S\subseteq Q_n$, is reachable if $|S|\le k$.
\smallskip

\noin
{\bf Induction Step:} We prove that if $|S|=k+1$, then $\{p'\} \cup S$ is reachable.
Let  $S=\{q_0,q_1,\dots, q_k\}$, where $0 \le q_0< q_1<\dots <q_k \le n-1$.
Suppose $q\in S$.
By assumption, sets
$\{p'\}\cup (S\setminus \{q\} -(q-1))$  are reachable for all $p'\in Q'_{m-1}$.
\smallskip

\noin
$\bullet$	\emph{All sets of the form $\{0'\} \cup S$ are reachable.}\\
	Note that $1\notin (S\setminus \{q\} -(q-1))$.
	By assumption, $\{(m-2)'\} \cup (S \setminus \{q\} -(q-1))$ is reachable.
	By Remark~\ref{rem:add1}, $\{0',1\}\cup (S\setminus \{q\} -(q-1))$  is  reachable.

\be 

\item
If there is an odd state $q$ in $S$, then
$\{0',1\}\cup (S\setminus \{q\} -(q-1)) \stackrel{b^{q-1}}{\longrightarrow} \{0',q\}\cup (S\setminus \{q\}) =\{0'\} \cup S.$

\item If there is no odd state in $S$ and $n$ is odd, then $S \subseteq \{0,2,\dots, n-1\}$. Pick $q \in S$. Then 
$\{0', 1\} \cup (S \setminus \{q\} - (q-1)) \xrightarrow{b^q} \{0', q+1\} \cup (S \setminus\{q\}+1) \xrightarrow{b^{n-1}} \{0', q\} \cup S \setminus\{q\} = \{0'\} \cup S$.

\item
If there is no odd state and $n$ is even, then $S\subseteq \{0,2,\dots,n-2\}$ (so $n-1\notin S$).
		
	\be

	\item
	If $0 \notin S$, then $0,1 \notin S+1$. By 1, $\{0'\} \cup (S+1)$ is reachable, since $S+1$ contains an odd 			state. Then 
	$\{0'\} \cup (S+1) \stackrel{a}{\longrightarrow} \{1'\} \cup (S+1) \stackrel{b^{n-1}}{\longrightarrow} \{0'\} \cup S$.
			
	\item
	If $2\notin S$, then $0,1 \notin S-1$. By 1, $\{0'\} \cup (S-1)$ is reachable, since $S-1$ contains an odd state. Then 
	$\{0'\} \cup (S-1) \stackrel{a}{\longrightarrow} \{1'\} \cup (S-1) \stackrel{b^{n+1}}{\longrightarrow} \{0'\} \cup S$.

	\item
	If $\{0,2\} \subseteq S$, then $0\notin S-1$, and $1, n-1\in S-1$. By 1, $\{0'\} \cup (S-1)$ is reachable, since $1\in 		S-1$. 
	Note that $aba$ sends 1 to 0, $n-1$ to 1, and adds 1 to each state $q\ge 3$ of $S-1$; thus $2\notin (S-1)aba$, and 
	$\{0'\} \cup (S-1) \stackrel{aba}{\longrightarrow} \{1', 0, 1\} \cup S\setminus \{0,2\}$.
	Next, $b^{n-1}$ sends 0 to $n-1$ and subtracts 1 from every other element of $S\setminus \{0,2\}$.
	Hence
	$\{1', 0, 1\} \cup S\setminus \{0,2\} \stackrel{b^{n-1}}{\longrightarrow} \{0',n-1,0\} \cup (S	\setminus \{0,2\} -1) 	 \stackrel{ab }{\longrightarrow}  \{0',0,2\} \cup (S	\setminus \{0,2\}) =\{0'\} \cup S$.

	\ee
		
\ee		
				
\noin
$\bullet$ \emph{All sets of the form $\{1'\} \cup S$ are reachable.}\\
If 0 and 1 are not in $S$ or are both in $S$, then $\{0'\} \cup S \stackrel{a}{\longrightarrow} \{1'\} \cup S$.
If $0\in S$ but $1\notin S$, then $\{0',1\} \cup S\setminus\{0\} \stackrel{a}{\longrightarrow} \{1',0\} \cup S\setminus\{0\} =\{1'\} \cup S$.
If $1\in S$ but $0\notin S$, then $\{0',0\} \cup S\setminus\{1\} \stackrel{a}{\longrightarrow} \{1',1\} \cup S\setminus\{1\} =\{1'\} \cup S$.
\medskip

\noin
\label{p'}
$\bullet$ \emph{All sets of the form $\{p'\} \cup S$, where $2 \le p\le m-2$,  are reachable.}\\
If $p$ is even, then $\{0'\} \cup S \stackrel{a^p}{\longrightarrow} \{p'\} \cup S$.\\
If $p$ is odd, then $\{1'\} \cup S \stackrel{a^{p-1}}{\longrightarrow} \{p'\} \cup S$.
\medskip

\noin 
{\bf Claim 2:}
\label{(m-1)'}
\emph{All sets of the form $\{(m-1)', 0\} \cup S$ are reachable.}
	\be
	\item
	By Claim~1, $\{(m-3)'\} \cup S$ is reachable.
	If $q_0=1$, then\\
	$\{(m-3)',1\} \cup S\setminus \{1\} \stackrel{a^2}{\longrightarrow} \{(m-1)', 0,1\} \cup S\setminus\{1\}=
	\{(m-1)', 0\} \cup S.$
	\item
	By Claim~1, $\{(m-2)'\} \cup S$ is reachable.
	If $q_0\ge 2$, then \\
	$\{(m-2)'\} \cup S \stackrel{a}{\longrightarrow} \{(m-1)', 0\} \cup S.$
	\ee
	\smallskip

\noin 
 {\bf Claim 3:}
\emph{All sets of the form $S$ are reachable.}\\
By~Claim 1,  $\{0'\} \cup S$ is reachable for every $S$, and $\{0'\} \cup S \stackrel{d}{\longrightarrow} S.$
\smallskip

For distinguishability, note that only state $q$ accepts $w_q = b^{n-1-q}$ in $\cD_n$. Hence, if two states of the product have different  sets $S$ and $S'$ and $q\in S\oplus S'$, then they can be distinguished by $w_q$. State $\{p'\} \cup S$ is distinguished from $S$ by $ca^{m-1-p}b^{n-1}$. If $p <q$, states $\{p'\} \cup S$ and $\{q'\} \cup S$ are distinguished as follows. Use $ca^{m-1-q}$ to reach $\{(p+m-1-q)'\}$ from $p'$ and $\{(m-1)'\} \cup \{0\}$ from $q'$. The reached states are distinguishable since they differ in their subsets of $Q_n$. \qed
\end{proof}

\section{Most Complex Regular Languages}

A \emph{most complex} regular language stream is one that, together with some dialects, meets 
 the complexity bounds for all boolean operations, product, star, and reversal, and has the largest syntactic semigroup and most complex atoms\footnote{The \emph{atom congruence} is a left congruence defined as follows: two words $x$ and $y$ are equivalent if 
 $ux\in L$ if and only if  $uy\in L$ for all $u\in \Sig^*$. 
 Thus $x$ and $y$ are equivalent if
 $x\in u^{-1}L$ if and only if $y\in u^{-1}L$.
 An equivalence class of this relation is called an \emph{atom} of $L$~\cite{BrTa14,Iva16}. 
It follows that an atom is a non-empty intersection of complemented and uncomplemented quotients of $L$. 
The number of atoms and their quotient complexities are possible measures of complexity of regular languages~\cite{Brz13}.
For more information about atoms and their complexity, see~\cite{BrTa13,BrTa14,Iva16}.
}~\cite{Brz13}. 
A most complex stream should have the smallest possible alphabet sufficient to meet all the bounds.
Most complex streams are useful in systems dealing with regular languages and finite automata. One would like to know the maximal sizes of automata that can be handled by the system. In view of the existence of most complex streams, one stream can be used to test all the operations.
Here we use the stream of~\cite{Brz13}  shown in Figure~\ref{fig:RegWit}. 

\begin{theorem}[Most Complex Regular Languages]
\label{thm:main}
For each $n\ge 3$, the DFA of Definition~\ref{def:regular} is minimal and its 
language $L_n(a,b,c,d)$ has complexity $n$.
The stream $(L_m(a,b,c,d) \mid m \ge 3)$  with dialect streams
$(L_n(a,b,-,c) \mid n \ge 3)$ and $(L_n(b,a,-,d) \mid n \ge 3)$
is most complex in the class of regular languages.
In particular, it meets all the complexity bounds below, which are maximal for regular languages.
In several cases the bounds can be met with a restricted alphabet.
\begin{enumerate}
\item
The syntactic semigroup of $L_n(a,b,c)$ has cardinality $n^n$.  
\item
Each quotient of $L_n(a)$ has complexity $n$.
\item
The reverse of $L_n(a,b,c)$ has complexity $2^n$, and $L_n(a,b,c)$ has $2^n$ atoms.
\item
For each atom $A_S$ of $L_n(a,b,c)$, the complexity $\kappa(A_S)$ satisfies:
\begin{equation*}
	\kappa(A_S) =
	\begin{cases}
		2^n-1, 			& \text{if $S\in \{\emp,Q_n\}$;}\\
		1+ \sum_{x=1}^{|S|}\sum_{y=1}^{n-|S|} \binom{n}{x}\binom{n-x}{y},
		 			& \text{if $\emp \subsetneq S \subsetneq Q_n$.}
		\end{cases}
\end{equation*}
\item
The star of $L_n(a,b)$ has complexity $2^{n-1}+2^{n-2}$.
\item
The product $L'_m(a,b,-,c) L_n(b,a,-,d)$ has complexity $m2^n+2^{n-1}$.
\item
The complexity of $L'_m(a,b,-,c) \circ L_n(b,a,-,d)$ 
is $mn+m+n+1$ if $\circ\in \{\cup,\oplus\}$,
that of  $L'_m(a,b,-,c) \setminus L_n(b,a)$ is 
 $mn+m$, and 
 that of $L'_m(a,b) \cap L_n(b,a)$ is $mn$. 

\end{enumerate}
\end{theorem}
\begin{proof}
The proofs of 1--5 can be found in~\cite{Brz13}, and
Claims 6 and 7 are proved in the present paper, Theorems~\ref{thm:boolean} and~\ref{thm:product}.
\qed
\end{proof}
\begin{proposition}[Marek Szyku{\l}a, personal communication]
At least four inputs are required for a most complex regular language. In particular, four inputs are needed for union: two inputs are needed to reach all pairs of states in $Q'_m \times Q_n$, one input in $\Sigma' \setminus \Sigma$ for pairs $(p',\emp)$ with $p'\in Q'_m$, and one in $\Sigma \setminus \Sigma'$ for pairs $(\emp',q)$ with $q \in  Q_n$.
\end{proposition}

\section{Conclusions}

Two complete DFAs over different alphabets $\Sig'$ and $\Sig$ are incomplete DFAs over  $\Sig' \cup \Sig$. Each DFA can be completed by adding  an empty state and sending all transitions induced by letters not in the DFA's alphabet to that state.
This results in an $(m+1)$-state DFA and an $(n+1)$-state DFA. 
From the theory about DFAs over the same alphabet we know that $(m+1)(n+1)$ is an upper bound for all boolean operations on the original DFAs, and that $m2^{n+1} +2^n$ is an upper bound for product. 
We have shown that the tight bounds for boolean operations are
$(m+1)(n+1)$ for union and symmetric difference, $mn+m$ for difference, and $mn$ for intersection, while
the tight bound for product is $m2^n+2^{n-1}$. In the same-alphabet case the tight bound is $mn$ for all boolean operations and it is $(m-1)2^n+2^{n-1}$ for product. 
In summary,  the restriction of identical alphabets is unnecessary and leads to incorrect results. 

It should be noted that if the two languages in question already have empty quotients, then making the alphabets the same does not require the addition of any states, and the traditional same-alphabet methods are correct. This is the case, for example, for prefix-free, suffix-free and finite languages.
\medskip

\noin
{\bf  Acknowledgment}
I am very grateful to Sylvie Davies, Bo Yang Victor Liu and Corwin Sinnamon for careful proofreading and constructive comments.
I thank Marek Szyku{\l}a for contributing Proposition~1 and several other useful comments.


\begin{thebibliography}{10}
\providecommand{\url}[1]{\texttt{#1}}
\providecommand{\urlprefix}{URL }

\bibitem{BBMR14}
Bell, J., Brzozowski, J., Moreira, N., Reis, R.: Symmetric groups and quotient
  complexity of boolean operations. In: Esparza, J., et~al. (eds.) ICALP 2014.
  LNCS, vol. 8573, pp. 1--12. Springer (2014)

\bibitem{Brz10a}
Brzozowski, J.: Quotient complexity of regular languages. J. Autom. Lang. Comb.
   15(1/2),  71--89 (2010)

\bibitem{Brz13}
Brzozowski, J.: In search of the most complex regular languages. Int. J. Found.
  Comput. Sc.  24(6),  691--708 (2013)

\bibitem{BrTa13}
Brzozowski, J., Tamm, H.: Complexity of atoms of regular languages. Int. J.
  Found. Comput. Sc.  24(7),  1009--1027 (2013)

\bibitem{BrTa14}
Brzozowski, J., Tamm, H.: Theory of \'atomata. Theoret. Comput. Sci.  539,
  13--27 (2014)

\bibitem{GSY11}
Gao, Y., Salomaa, K., Yu, S.: Transition complexity of incomplete \mbox{DFAs}.
  Fund. Inform.  110,  143--158 (2011)

\bibitem{Iva16}
Iv\'an, S.: Complexity of atoms, combinatorially. Inform. Process. Lett.
  116(5),  356--360 (2016)

\bibitem{MMR15}
Maia, E., Moreira, N., Reis, R.: Incomplete operational transition complexity
  of regular languages. Inform. and Comput.  244,  1--22 (2015)

\bibitem{Mas70}
Maslov, A.N.: Estimates of the number of states of finite automata. Dokl. Akad.
  Nauk SSSR  194,  1266--1268 (Russian). (1970), \mbox{English translation}:
  Soviet Math. Dokl. {\bf 11} (1970) 1373--1375

\bibitem{Yu01}
Yu, S.: State complexity of regular languages. J. Autom. Lang. Comb.  6,
  221--234 (2001)

\bibitem{YZS94}
Yu, S., Zhuang, Q., Salomaa, K.: The state complexities of some basic
  operations on regular languages. Theoret. Comput. Sci.  125,  315--328 (1994)

\end{thebibliography}

\providecommand{\noopsort}[1]{}

\end{document}